\documentclass[10pt]{article}
\usepackage{amsmath}
\usepackage{amssymb}
\usepackage{amsthm}
\usepackage{mathrsfs}
\usepackage{tikz,authblk}
\newtheorem{theorem}{Theorem}[section]
\newtheorem{lemma}[theorem]{Lemma}

\theoremstyle{definition}

\begin{document}
\title{Seymour's second neighbourhood conjecture for
quasi-transitive oriented graphs\thanks{Research of RL was partially supported by NNSFC under no. 11401353 and TYAL of Shanxi.}}
\author[1]{Gregory Gutin}
\author[2]{Ruijuan Li}
\affil[1]{Department of Computer Science, Royal Holloway, University of London, TW20 0EX Egham, UK}
\affil[2]{School of Mathematical Sciences, Shanxi University, Taiyuan, Shanxi, 030006, PR China}
\maketitle

\begin{abstract}
Seymour's second neighbourhood conjecture asserts that every oriented graph has a vertex whose second out-neighbourhood is at least as large as its out-neighbourhood. In this paper, we prove that the conjecture holds for
quasi-transitive oriented graphs, which is a superclass of tournaments and transitive acyclic digraphs. A digraph $D$ is called quasi-transitive is for every pair $xy,yz$ of arcs between distinct vertices $x,y,z$, $xz$ or $zx$ (``or'' is inclusive here) is in $D$.
\end{abstract}




\section{Introduction}

For convenience of the reader we provide all necessary terminology and notation in one section, Section 2.

One of the most interesting and challenging open questions concerning digraphs is Seymour's Second Neighbourhood Conjecture (SSNC) \cite{DeanLatka}, which asserts that one can always find, in an oriented graph $D$, a vertex $x$ whose second out-neighbourhood is at least as large as its out-neighbourhood, i.e.  $|N^+(x)| \le |N^{++}(x)|$. Following \cite{CohnGodboleHarknessZhang}, we will call such a vertex $x$ a {\em Seymour vertex}.

Observe that SSNC is not true for digraphs in general. Consider $\overleftrightarrow{K}_n$, the complete digraph on $n$ vertices. For each vertex $v\in V(\overleftrightarrow{K}_n)$, $N_{\overleftrightarrow{K}_n}^+(v)=V(\overleftrightarrow{K}_n)\setminus\{v\}$ while $N_{\overleftrightarrow{K}_n}^{++}(v)=\emptyset$.
The conjecture trivially holds for digraphs $D$ which contain a vertex of out-degree zero, e.g. for
acyclic digraphs. Indeed, $N^+_D(v_n)=N^{++}_D(v_n)=\emptyset$.

The first non-trivial result for SSNC was obtained by Fisher \cite{Fisher} who proved Dean's conjecture \cite{DeanLatka}, which is SSNC restricted to tournaments. Fisher used Farkas' Lemma and averaging arguments.
\begin{theorem}\cite{Fisher}\label{2nd:T}
In any tournament $T$, there is a vertex $v$ such that $|N_T^+(v)|\le |N_T^{++}(v)|$.
\end{theorem}
 A more elementary proof of SSNC for tournaments was given by Havet and Thomass\'{e} \cite{HavetThomasse} who introduced a median order approach. Their proof also yields the following stronger result.
 \begin{theorem}\cite{HavetThomasse}\label{2nd:2T}
 A tournament $T$ with no vertex of out-degree zero has at least two vertices $v$ such that $|N_T^+(v)|\le |N_T^{++}(v)|$.
 \end{theorem}
Fidler and Yuster \cite{FidlerYuster} further developed the median order approach and proved that SSNC holds for oriented graphs $D$ with minimum degree $|V(D)|-2$, tournaments minus a star, and tournaments minus the arc set of a subtournament.
The median order approach was also used by Ghazal \cite{Ghazal} who proved a weighted version of SSNC for tournaments missing a generalized star.
Kaneko and Locke \cite{KanekoLocke} proved SSNC for oriented graphs with minimum out-degree at most 6.   Cohn, Godbole, Wright Harkness, and Zhang \cite{CohnGodboleHarknessZhang} proved that the conjecture holds for random oriented graphs.

Another approach to SSNC is to determine the maximum value $\gamma$ such that in every oriented graph $D$, there exists a vertex $x$ such that $|N_D^+(x)|\le \gamma|N_D^{++}(x)|$. SSNC asserts that $\gamma=1$. Chen, Shen, and Yuster \cite{ChenShenYuster} proved that $\gamma\ge r$ where $r=0.657298\ldots$ is the unique real root of $2x^3+x^2-1=0$. They also claim a slight improvement to $r\ge 0.67815\ldots$.

In this paper, we consider Seymour's Second Neighbourhood Conjecture for quasi-transitive digraphs. We use a decomposition theorem of Bang-Jensen and Huang \cite{qtdschar} for quasi-transitive digraphs, Theorem \ref{qtdschar}. We also use some structural properties of extended tournaments, a subclass of quasi-transitive oriented graphs.


\section{Terminology and Notation}

We will assume that the reader is familiar with the standard terminology on digraphs and refer to \cite{BangGutin} for terminology not discussed here. In this paper, all digraphs have no multiple arcs or loops.

We denote the vertex set and the arc set of a digraph $D$ by $V(D)$ and $A(D)$, respectively. For a vertex subset $X$, we denote by $D\langle X\rangle$ the subdigraph of $D$ induced by $X$, $D\langle V(D)-X\rangle$ by $D-X$. In addition, $D-x=D-\{x\}$ for a vertex $x$ of $D$.

Let $x,y$ be distinct vertices in $D$. If there is an arc from $x$ to $y$ then we say that $x$ {\em dominates} $y$, write $x \rightarrow y$ and call $y$ (respectively, $x$) an {\em out-neighbour} (respectively, an {\em in-neighbour}) of $x$ (respectively, $y$). For a subdigraph or simply a vertex subset $H$ of $D$ (possibly, $H=D$), we let $N_H^+(x)$ (respectively, $N_H^-(x)$) denote the set of out-neighbours (respectively, the set of in-neighbours) of $x$ in $H$ and call it {\em out-neighbourhood} (respectively, {\em in-neighbourhood}) of $x$ in $H$. Furthermore, $d_H^+(x)=|N_H^+(x)|$ (respectively, $d_H^-(x)=|N_H^-(x)|$) is called the {\em out-degree} (respectively, {\em in-degree}) of $x$. Let
\[
N_H^{++}(x)=\bigcup\limits_{u\in N_H^+(x)}N_H^+(u)\setminus N_H^+(x),
\]
which is called the {\em second out-neighbourhood} of $x$ in $H$.

A digraph $D$ is said to be {\em strong}, if for every pair of vertices $x$ and $y$, $D$ contains a directed path from $x$ to $y$ and a directed path from $y$ to $x$. A strong component of a digraph $D$ is a maximal induced subdigraph of $D$ which is strong. If $D_1,\ldots, D_t$ are the strong components of $D$, then clearly $V(D_1)\cup \ldots \cup V(D_t)=V(D)$ (a digraph with only one vertex is strong). Moreover, we must have $V(D_i)\cap V(D_j)=\emptyset$ for every $i\ne j$. The strong components of $D$ can be labelled $D_1,\ldots, D_t$ such that there is no arc from $D_j$ to $D_i$ unless $j<i$. We call such an ordering an {\em acyclic ordering
of the strong components} of $D$.

A digraph $D$ is {\em acyclic} if it has no directed cycle. An ordering $v_1, v_2, \ldots, v_n$ of vertices of a digraph $D$ is called {\em acyclic} if for every arc $v_iv_j\in A(D)$, we have $i<j$. It is well-known that every acyclic digraph has an acyclic ordering \cite{BangGutin}. Clearly, an acyclic ordering is a median order for acyclic digraphs.

Let $D$ be a digraph with vertex set $\{v_1, v_2, \ldots, v_n\}$, and let $G_1, G_2, \ldots, G_n$ be digraphs which are pairwise vertex disjoint. The {\em composition} $D[G_1, G_2, \ldots, G_n]$ is the digraph $L$ with vertex set $V(G_1)\cup V(G_2)\cup\ldots \cup V(G_n)$ and arc set $(\cup_{i=1}^n A(G_i))\cup\{g_ig_j\, |\, g_i\in V(G_i), g_j\in V(G_j), v_iv_j\in A(D)\}$. If $D=H[S_1, S_2,$ $\ldots, S_h]$ and none of the digraphs $S_1, S_2, \ldots, S_h$ has an arc, then $D$ is an {\em extension} of $H$. For $i\in \{1,2,\ldots, s\}$, each $S_i$ called the {\em partite set} of $D$.

An {\em oriented graph} is a digraph with no cycle of length two. A {\em tournament} is an oriented graph where every pair of distinct vertices are adjacent. An {\em extended tournament} is an extension of a tournament.

A digraph $D$ is {\em quasi-transitive} if for every pair $xy$ and $yz$ of arcs in $D$ with $x\ne z$ implies that $x$ and $y$ are adjacent. A digraph $D$ is {\em transitive} if, for every pair $xy$ and $yz$ of arcs in $D$ with $x\ne z$, the arc $xz$ is also in $D$. Observe that each transitive digraph is quasi-transitive and each extended tournament is also quasi-transitive.

To make quasi-transitive digraphs easier to deal with, Bang-Jensen and Huang \cite{qtdschar} introduced the following characterization of this class of digraphs.

\begin{theorem}\cite{qtdschar}\label{qtdschar}
Let $D$ be a quasi-transitive digraph.
\begin{itemize}
\item If $D$ is not strong, then there exists a transitive oriented graph $T$ with vertices $\{u_1, u_2,\ldots,u_t\}$ and strong quasi-transitive digraphs $H_1, H_2, \ldots,H_t$ such that $D = T[H_1, H_2,\ldots, H_t]$, where $H_i$ is substituted for $u_i, i\in \{1, 2, \ldots, t\}$.
\item If $D$ is strong, then there exists a strong semicomplete digraph $S$ with vertices $\{v_1, v_2, \ldots, v_s\}$ and quasi-transitive digraphs $Q_1, Q_2, \ldots, Q_s$ such that $Q_i$ is either a vertex or is non-strong and $D = S[Q_1, Q_2,\ldots, Q_s]$, where $Q_i$ is subsituted for $v_i, i\in\{1, 2, \ldots, s\}$.
\end{itemize}
\end{theorem}

The decomposition described in Theorem \ref{qtdschar} is called the {\em canonical decomposition} of the quasi-transitive digraph $D$.



\section{Main Results}

We give the following easy but useful observation, which indicates the relationship between the Seymour vertex of a quasi-transitive oriented graph and the one of an extended tournament.

\begin{lemma}\label{qt&et}
Let $D$ be a strong quasi-transitive oriented graph and $D = S[Q_1,$ $Q_2,\ldots, Q_s]$ be the canonical decomposition. Let $D^*=S[V_1, V_2,\ldots, V_s]$ be an extended tournament, where $V_i$ is the vertex set of the subdigraph $Q_i$ for $i\in\{1, 2, \ldots, s\}$.  If there is a vertex $x\in V_i$ such that $x$ is a Seymour vertex of $Q_i$ and a Seymour vertex of $D^*$, then $x$ is a Seymour vertex of $D$.
\end{lemma}
\begin{proof}
Since $x$ is a Seymour vertex in $Q_i$ and a Seymour vertex in $D^*$, we have
\[
|N^+_{Q_i}(x)| \le |N^{++}_{Q_i}(x)|, \quad |N^+_{D^*}(x)| \le |N^{++}_{D^*}(x)|.
\]
Clearly,
\[
N^+_{D}(x)=N^+_{Q_i}(x)\cup N^+_{D^*}(x), \quad N^{++}_{D}(x)=N^{++}_{Q_i}(x)\cup N^{++}_{D^*}(x)
\]
Thus $|N^+_{D}(x)|\le | N^{++}_{D}(x)|$.
\end{proof}

First we deal with SSNC for extended tournaments. 

\begin{theorem}\label{2nd:ET}
Let $D$ be an extended tournament. Then there is a vertex $v$ such that $|N_D^+(v)|\le|N_D^{++}(v)|$.
\end{theorem}
\begin{proof}
Let $D=S[V_1, V_2,\ldots, V_s]$ be an extended tournament with each $V_i$ being an independent set. Now replace each $V_i$ with a transitive tournament on the same vertex set $V_i$ and obtain a new digraph $D'$. Clearly, $D'$ is a tournament hence satisfies the SSNC with some vertex $v$. Note that $|N^+_{D}(v)|\le |N^+_{D'}(v)|$ and $|N^{++}_{D}(v)|=|N^{++}_{D'}(v)|$. Thus $v$ is a Seymour vertex in $D$.
\end{proof}

The following theorem shows that we can generalize Theorem \ref{2nd:T} to quasi-transitive oriented graphs.

\begin{theorem}\label{2nd:QTOG}
In any quasi-transitive oriented graph $D$, there is a vertex $v$ such that $|N_D^+(v)|\le|N_D^{++}(v)|$.
\end{theorem}
\begin{proof}
The proof is by induction on the order $n$ of $D$. It is easy to check that the cases $1\le n\le 3$ hold. Assume that $n\ge 4$.

{\it Case 1:} $D$ is not strong. Let $D=T[H_1, H_2, \ldots, H_t]$ be the canonical decomposition of $D$, where $T$ is a transitive oriented graph and $H_i$ is a strong quasi-transitive oriented graph for $i\in\{1,2,\ldots, t\}$. Without loss of generality, assume that $H_1, H_2, \ldots, H_t$ is the acyclic ordering of the strong components of $D$. By induction hypothesis, let $v$ be a Seymour vertex of $H_t$. This means $|N^+_{H_t}(v)|\le |N^{++}_{H_t}(v)|$. Clearly, $N^+_D(v)=N^+_{H_t}(v)$ and $N^{++}_D(v)=N^{++}_{H_t}(v)$. Thus $|N^+_D(v)|\le |N^{++}_D(v)|$.

{\it Case 2:} $D$ is strong. Let $D=S[Q_1, Q_2, \ldots, Q_s]$ be the canonical decomposition of $D$, where $S$ is a strong tournament and $Q_i$ is a single vertex or non-strong quasi-transitive oriented graph for $i\in\{1,2,\ldots, s\}$. Let $D^*=S[V_1, V_2,\ldots, V_s]$ be an extended tournament, where $V_i$ is the vertex set of the subdigraph $Q_i$ for $i\in\{1, 2, \ldots, s\}$. By Theorem \ref{2nd:ET}, $D^*$ has a Seymour vertex $v$. Assume $v\in V_i$. Then each vertex in $V_i$ is a Seymour vertex in $D^*$. By induction hypothesis, there is a Seymour vertex of $Q_i$, say also $v$. By Lemma \ref{qt&et}, $v$ is a Seymour vertex in $D$ and $|N^+_D(v)|\le |N^{++}_D(v)|$.
\end{proof}

Now we generalize Theorem \ref{2nd:2T} to extended tournaments. The following theorem indicates that an extended tournament always has two vertices with large second out-neighbourhood, provided that every vertex has out-degree at least 1 and, the second out-neighbourhood of Seymour vertex is more than out-neighbourhood if such two Seymour vertices are in a same partite set.

\begin{theorem}\label{2nd:2ET}
Let $D=S[V_1, V_2,\ldots, V_s]$ be an extended tournament with each $V_i$ being an independent set. If $D$ has no vertex of out-degree zero, then
\begin{itemize}
\item[(a)] there are at least two vertices $v$ such that $|N_D^+(v)|\le|N_D^{++}(v)|$, and
\item[(b)] there exists at least one vertex $v$ such that $|N^+_D(v)|<|N^{++}_D(v)|$ unless there is another Seymour vertex  $u$ which is in a distinct partite set from $v$.
\end{itemize}
\end{theorem}
\begin{proof}
Let $D=S[V_1, V_2,\ldots, V_s]$ be an extended tournament. Now replace each $V_i$ with a transitive tournament on the same vertex set $V_i$. Now $D$ becomes a tournament, say $T$. Since $D$ has no vertex of out-degree zero, so does $T$. By Theorem 1.2, $T$ has at least two Seymour vertices, say $u,v$.
If $u,v$ are in the different partite sets, there is nothing to do. So assume that $u,v\in V_i$ for some $i\in\{1,2,\ldots,s\}$. Without loss of generality, $v$ dominates $u$ in $T$. Note that $|N_T^+(v)|\le|N_T^{++}(v)|$. Thus $|N_D^+(v)|<|N_D^{++}(v)|$.
Now we show that a quasi-transitive oriented graph always has two vertices with large second out-neighbourhood, provided that every vertex has out-degree at least 1.
\end{proof}

Finally, we generalize Theorem \ref{2nd:2T} to quasi-transitive oriented graphs.

\begin{theorem}\label{2nd:2QTOG}
A quasi-transitive oriented graph $D$ with no vertex of out-degree zero has at least two vertices $v$ such that $|N_D^+(v)|\le |N_D^{++}(v)|$.
\end{theorem}
\begin{proof}
The proof is by induction on the order $n$ of $D$. It is easy to check that the case $n=3$ holds. Assume that $n\ge 4$.

{\it Case 1:} $D$ is not strong. Let $D=T[H_1, H_2, \ldots, H_t]$ be the canonical decomposition of $D$, where $T$ is a transitive oriented graph and $H_i$ is a strong quasi-transitive oriented graph for $i\in\{1,2,\ldots, t\}$. Without loss of generality, assume that $H_1, H_2, \ldots, H_t$ is the acyclic ordering of the strong components of $D$. Since $D$ has no vertex out-degree zero, the last component $H_t$ must contain at least three vertices. This means $H_t$ is a quasi-transitive oriented graph with no vertex of out-degree zero. By induction hypothesis, there are at least two Seymour vertices in $H_t$. Since every Seymour vertex of $H_t$ is also a Seymour vertex of $D$, $D$ has at least two vertices $v$ such that $|N_D^+(v)|\le |N_D^{++}(v)|$.

{\it Case 2:} $D$ is strong. Let $D=S[Q_1, Q_2, \ldots, Q_s]$ be the canonical decomposition of $D$, where $S$ is a strong tournament and $Q_i$ is a single vertex or non-strong quasi-transitive oriented graph for $i\in\{1,2,\ldots, s\}$. Let $D^*=S[V_1, V_2,\ldots, V_s]$ be an extended tournament, where $V_i$ is the vertex set of the subdigraph $Q_i$ for $i\in\{1, 2, \ldots, s\}$. Clearly, $D^*$ is strong and hence has no vertex of out-degree zero. Let $\sigma=(v_1,v_2,\ldots,v_n)$ be a well-organized median order of $D^*$. By Theorem \ref{2nd:2ET}(b), $|N^+_{D^*}(v_n)|<|N^{++}_{D^*}(v_n)|$ unless there is another Seymour vertex $u$ which is in a distinct partite set from $v_n$. For the case when the latter holds, $D^*$ has two Seymour vertices which belong to different partite sets, say $V_\alpha$ and $V_\beta$. By induction hypothesis, there is a Seymour vertex in each $Q_i$ for $i\in \{1,2,\ldots, s\}$. Now Theorem \ref{qt&et} implies that the Seymour vertices of $Q_\alpha$ and $Q_\beta$ are also Seymour vertices of $D$.

So assume that $|N^+_{D^*}(v_n)|<|N^{++}_{D^*}(v_n)|$. For convenience, assume $v_n\in V_1$. We claim that the partite set $V_1$ contains at least two vertices. Indeed, if not, then $v_n$ is the unique vertex of $V_1$. By Theorem \ref{2nd:2ET}(a), there must exist another Seymour vertex $u$ which is not in $V_1$. As shown above, $u$ is also a Seymour vertex of $D$. So $V_1$ contains at least two vertices.

If there are at least two Seymour vertices in $Q_1$, then they are also Seymour vertices in $D$. So assume that $Q_1$ has exactly one Seymour vertex, say $v_n$. Now we claim that there is another vertex $u\in V_1$ distinct from $v_n$ such that $|N_{Q_1}^+(u)|-1\le |N_{Q_1}^{++}(u)|$. In fact,
set $Q_1=T_1[Q_1^1,Q_1^2,\ldots, Q_1^r]$ be the canonical decomposition of $Q_1$, where $T_1$ is a transitive oriented graph and $Q_1^i$ is a strong quasi-transitive oriented graph for $i\in\{1,2,\ldots, r\}$. Also, assume that $Q_1^1,Q_1^2,\ldots, Q_1^r$ is the acyclic ordering of the strong components of $Q_1$. Clearly, $Q_1^r$ is the unique terminal strong component and $v_n$ is the unique vertex of $Q_1^r$. By induction hypothesis, there is a Seymour vertex $u$ in $Q_1^{r-1}$. This means $|N_{Q_1^{r-1}}^+(u)|\le |N_{Q_1^{r-1}}^{++}(u)|$. Now
\[
|N_{Q_1}^+(u)|-1=|N_{Q_1^{r-1}}^+(u)|\le |N_{Q_1^{r-1}}^{++}(u)|\le |N_{Q_1}^{++}(u)|.
\]
Since $u$ and $v_n$ are in the same partite set $V_1$ of $D^*$, the inequality $|N^+_{D^*}(u)|<|N^{++}_{D^*}(u)|$ holds. Clearly, $N^+_D(u)=N_{Q_1}^+(u)\cup N^+_{D^*}(u)$ and $N^{++}_D(u)=N_{Q_1}^{++}(u)\cup N^{++}_{D^*}(u)$. Thus $|N^+_D(u)|\le |N^{++}_D(u)|$ and the theorem holds.
\end{proof}


\end{document}